\documentclass[letterpaper, 10 pt, conference]{ieeeconf}  % Comment this line out if you need a4paper

\IEEEoverridecommandlockouts                              % This command is only needed if 
\usepackage{graphicx} % for pdf, bitmapped graphics files
\usepackage{amsmath} % assumes amsmath package installed
\DeclareMathOperator*{\argmax}{arg\,max}
\usepackage{amssymb}  % assumes amsmath package installed
\usepackage{algorithm}
\usepackage{algorithmic}
\usepackage{tikz}
\usepackage{color}
\usepackage{easy-todo}

\newtheorem{remark}{Remark}
\newtheorem{theorem}{Theorem}

\title{\LARGE \bf
Data-Driven Distributed Stochastic Model Predictive Control with Closed-Loop Chance Constraint Satisfaction
}

\author{Simon Muntwiler*, Kim P. Wabersich*, Lukas Hewing, and Melanie N. Zeilinger% <-this % stops a space
\thanks{*The first two authors contributed equally to this work.}
\thanks{This work was supported by the Swiss National Science Foundation under grant no. PP00P2\textunderscore157601/1. Simon Muntwiler's research was supported by funds from the Bosch Research Foundation im Stifterverband.
}% <-this % stops a space
\thanks{The authors are members of the Institute for Dynamic Systems and
	Control, ETH Z\"urich, Z\"urich, Switzerland {\tt\small\{simonmu,wkim,lhewing,} {\tt\small mzeilinger\}@ethz.ch}.
}%
}

\begin{document}

\setlength{\textfloatsep}{15pt}

\maketitle
\thispagestyle{empty}
\pagestyle{empty}

\begin{tikzpicture}[overlay, remember picture]
\node[anchor=south,yshift=8pt] at (current page.south) {\fbox{\parbox{0.99\textwidth}{\footnotesize \textbf{Published in: 2021 European Control Conference (ECC). DOI: 10.23919/ECC54610.2021.9655214.}\\ 
\textcopyright \space 2022 IEEE. Personal use of this material is permitted. Permission from IEEE must be obtained for all other uses, in any current or future media, including reprinting/republishing this material for advertising or promotional purposes, creating new collective works, for resale or redistribution to servers or lists, or reuse of any copyrighted component of this work in other works.}}};
\end{tikzpicture}
%%%%%%%%%%%%%%%%%%%%%%%%%%%%%%%%%%%%%%%%%%%%%%%%%%%%%%%%%%%%%%%%%%%%%%%%%%%%%%%%
\begin{abstract}
Distributed model predictive control methods for uncertain systems often suffer from considerable conservatism and can tolerate only small uncertainties due to the use of robust formulations that are amenable to distributed design and optimization methods.
In this work, we propose a distributed stochastic model predictive control (DSMPC) scheme for dynamically coupled linear discrete-time systems subject to unbounded additive disturbances that are potentially correlated in time.
An indirect feedback formulation ensures recursive feasibility of the DSMPC problem, and a data-driven, distributed and optimization-free constraint tightening approach allows for exact satisfaction of chance constraints during closed-loop control, addressing typical sources of conservatism.
The computational complexity of the proposed controller is similar to nominal distributed MPC.
The approach is demonstrated in simulation for the temperature control of a large-scale data center subject to randomly varying computational loads.
\end{abstract}

%%%%%%%%%%%%%%%%%%%%%%%%%%%%%%%%%%%%%%%%%%%%%%%%%%%%%%%%%%%%%%%%%%%%%%%%%%%%%%%%
\section{INTRODUCTION}

Sensing and communication capabilities are increasingly available in many technical systems, allowing interconnected systems to measure information locally and share it with other agents to optimize a common global objective.
In the case of manufacturing systems, for instance, multiple machines may be necessary to assemble a product, and the use of each machine can be optimally scheduled based on shared information between the different production steps, increasing the overall efficiency.
Solving such a large-scale control problem in a centralized manner, however, often results in intractable communication requirements or computationally infeasible optimization problems \cite{Christofides2013}.
Distributed control algorithms address these issues by exploiting the distributed structure of the system and carry out computations locally while only requiring state information from neighboring subsystems.

For large-scale systems in particular, deriving an accurate system model and description of its operating conditions is a challenging task. 
In model predictive control (MPC), the resulting uncertainties are often modeled through additive disturbances acting on the system.
These can be addressed in a robust fashion, guaranteeing constraint satisfaction under \emph{all} disturbance realizations in a compact set~\cite{Bemporad1999}. 
Distributed robust approaches, however, tend to introduce conservatism as a result of, e.g., enforcing a distributed structure on a robust positive invariant set~\cite{Conte2013} or handling dynamic couplings as uncertainties~\cite{Jia2002}.
A promising alternative is a stochastic approach, where an underlying stochastic characteristic of the disturbances is considered and constraints are satisfied with a certain probability level \cite{Mesbah2016}, providing a quantified assessment of risk.
The closed-loop analysis for a stochastic MPC approach, however, is typically more challenging and often relies on bounded disturbance distributions~\cite{Lorenzen2017}, or unimodality and symmetric tightening assumptions~\cite{Hewing2018e,Mark2019,Mark2020}.
Existing distributed stochastic MPC (DSMPC) frameworks assume Gaussian disturbance distributions \cite{Mark2019} or general mean-variance information of i.i.d.\ disturbances \cite{Mark2020}.
In addition, chance constraints are usually enforced for all agents simultaneously, which can again introduce conservatism, in particular for large-scale systems.

\textit{Contributions:} This paper introduces a DSMPC scheme for dynamically coupled linear systems subject to additive non-i.i.d.\ disturbances with potentially unbounded support.
The goal is to regulate each local subsystem to its respective set-point, while
satisfying local chance constraints with a given probability level, and thereby ensuring the safety of each local subsystem.
Instead of assuming a given distribution of the disturbance, we only assume having access to samples of the disturbances either from experiments or simulations, resulting in a data-driven MPC formulation~\cite{Hewing2020}.
In contrast to existing DSMPC approaches, recursive feasibility of the proposed DSMPC optimization problem is ensured by relying on indirect feedback as introduced in \cite{Hewing2018}, where the actual measured state only enters the cost rather than the constraints.

A key contribution of the presented work is a data-driven and distributed tightening approach based on scenario optimization techniques~\cite{Campi2011,Hewing2019} to handle chance constraints on states and inputs as deterministic distributed constraints on nominal system states and inputs, while allowing us to provide guarantees for closed-loop chance constraint satisfaction for each individual agent.
In contrast to related schemes (e.g.,~\cite{Mark2019}), the constraint tightening does not introduce additional conservatism compared to a centralized solution and allows us to handle the local chance constraints in a non-conservative manner, meaning that the desired chance constraint probability level is realized exactly at time steps for which a given initial condition and any disturbance realization lead to the tightened constraint on the nominal system being active during closed-loop operation.
Moreover, the resulting optimization problem with respect to nominal states and inputs has computational complexity comparable to a nominal distributed MPC problem.

\textit{Related Work:} DSMPC algorithms based on distributional information on the disturbances have been introduced for linear systems with zero-mean i.i.d. additive Gaussian disturbances in \cite{Mark2019}, and extended to output-feedback in~\cite{Mark2020}.
These techniques similarly consider unbounded disturbance distributions, but cannot ensure recursive feasibility of the MPC problem directly; instead they make use of a recovery initialization in the case of infeasibility.
As a result, constraint satisfaction is guaranteed under symmetric tightening and unimodal disturbances only.
For bounded disturbances and dynamically decoupled systems with coupling constraints, approaches ensuring constraint satisfaction were presented in~\cite{Dai2016}.
A DSMPC framework based on disturbance samples, rather than distributional information, for linear systems with parametric uncertainty and additive disturbances was investigated in \cite{Rostampour2018}.
The constraints in the online DSMPC problem, however, need to be fulfilled for the entire set of disturbance samples, increasing the computational complexity.
Recursive feasibility and stability were not investigated.

\textit{Notation:} A stacked vector $v \in \mathbb{R}^n$ consisting of subvectors $v_i \in \mathbb{R}^{n_i}$ with $i \in \mathcal{M} \subseteq \mathbb{N}$ is denoted as $v = \mathrm{col}_{i \in \mathcal{M}}(v_i)$.
The distribution $\mathcal{Q}$ of a random variable $w$ is denoted as $w \sim \mathcal{Q}$, probabilities and conditional probabilities as $\mathrm{Pr}(A)$ and $\mathrm{Pr}(A|B)$ respectively.
By $\mathbb{E}_w(x)$ we denote the expected value of $x$ w.r.t. the random variable $w$.

\section{PROBLEM FORMULATION} \label{sec:preliminaries}

We consider a network of $M \in \mathbb{N}$ time-invariant coupled linear subsystems with discrete-time dynamics
\begin{equation}
x_i(t\!+\!1) = \left(\sum_{j=1}^{M}A_{ij}x_j(t)\right) + B_i u_i(t) + G_i w_i(t), \label{eq:local_system_dynamics}
\end{equation}
with local state $x_i(t) \in \mathbb{R}^{n_i}$, input $u_i(t) \in \mathbb{R}^{m_i}$ and stochastic disturbance $w_i(t) \in \mathbb{R}^{p_i}$ for each subsystem $i$ at time step $t$, where $A_{ij} \in \mathbb{R}^{n_i \times n_j}$, $B_i \in \mathbb{R}^{n_i \times m_i}$, and $G_i \in \mathbb{R}^{n_i \times p_i}$.
We denote the set of indices of all subsystems as $\mathcal{M}= \{1, \ldots, M\}$.
The set of neighbors $\mathcal{N}_i$ of subsystem~$i$ contains all indices of subsystems~$j$, for which $A_{ij}$ includes nonzero entries.
We assume that each subsystem is able to exchange information with all other subsystems in its neighborhood.
The local system dynamics of subsystem~$i$ can be written as
\begin{equation}
x_i(t\!+\!1) = A_{\mathcal{N}_i}x_{\mathcal{N}_i}(t) + B_i u_i(t) + G_i w_i(t),
\label{eq:neighborhood_system_dynamics}
\end{equation}
where $A_{\mathcal{N}_i} \in \mathbb{R}^{n_i \times n_{\mathcal{N}_i}}$ and $x_{\mathcal{N}_i}(t) = \mathrm{col}_{j \in \mathcal{N}_i}(x_j(t)) \in \mathbb{R}^{n_{\mathcal{N}_i}}$. 
Each subsystem $i$ is subject to $n_i^x$ half-space chance constraints on the local states and $n_i^u$ half-space chance constraints on the local inputs
\begin{subequations}
	\begin{align}
		\mathrm{Pr}(h_{i,j}^{x\top}x_i(t) \le 1) \ge p^x_{i,j}, \ j \in  \{1, \ldots, n_i^x\}, \\
		\mathrm{Pr}(h_{i,j}^{u\top}u_i(t) \le 1) \ge p^u_{i,j}, \ j \in  \{1, \ldots, n_i^u\},
	\end{align}\label{eq:constraints}%
\end{subequations}
where $h_{i,j}^x \in \mathbb{R}^{n_i}$, $h_{i,j}^u \in \mathbb{R}^{m_i}$ and the probabilities are understood conditioned on the initial state.

The objective is to control the distributed stochastic system over a potentially large, but finite, task horizon $\bar{N}$ while satisfying the chance constraints (\ref{eq:constraints}) at every time step $t$.
The stochastic disturbance sequence over the task horizon is assumed to be distributed according to $W = \left[\textrm{col}_{i\in\mathcal{M}}(w_i(0))^{\top},\ldots, \textrm{col}_{i\in\mathcal{M}}(w_i(\bar{N}))^{\top}\right]^{\top} \sim  \mathcal{Q}$, which can be a non-i.i.d.\ and correlated disturbance sequence with unbounded support.
It is not necessary to know the distribution of the disturbances, but we assume to have access to samples from the distribution over the entire task horizon.
Handling unbounded disturbances is especially important when the distribution and possibly existing bounds are not known in advance, with normal distributions as important special case.

In this paper, we introduce a distributed stochastic MPC scheme to approximate the solution of the optimal stochastic control problem by solving a simplified problem over a shorter horizon $N \ll \bar{N}$ in a receding horizon fashion.
The local system dynamics (\ref{eq:neighborhood_system_dynamics}) are split into a nominal state $z_i(t)$ and error $e_i(t)$ such that $x_i(t) = z_i(t) + e_i(t)$, as well as a nominal input $v_i(t)$ and potentially nonlinear tube controller $\pi_{i}(e_{\mathcal{N}_i}(t))$, resulting in
\begin{subequations}
	\begin{align} 
		z_i(t\!+\!1) &= A_{\mathcal{N}_i}z_{\mathcal{N}_i}(t) + B_i v_i(t), \label{eq:nominal_dynamics} \\
		e_i(t\!+\!1) &= A_{\mathcal{N}_i}e_{\mathcal{N}_i}(t) + B_i\pi_{i}(e_{\mathcal{N}_i}(t)) + G_iw_i(t), \label{eq:error_dynamics} \\
		x_i(t) &= z_i(t) + e_i(t), \\ 
		u_i(t) &= v_i(t) + \pi_{i}(e_{\mathcal{N}_i}(t)), \label{eq:split_dynamics_controller}
	\end{align} \label{eq:split_dynamics}%
\end{subequations}
with initial condition $z_i(0) = x_i(0)$, and therefore $e_i(0) = 0$, $z_{\mathcal{N}_i}(t) = \mathrm{col}_{j \in \mathcal{N}_i}(z_j(t)) \in \mathbb{R}^{n_{\mathcal{N}_i}}$ and $e_{\mathcal{N}_i}(t) = \mathrm{col}_{j \in \mathcal{N}_i}(e_j(t)) \in \mathbb{R}^{n_{\mathcal{N}_i}}$.
The MPC problem optimizes the nominal input $v_i(t)$, while the tube controller $\pi_{i}(e_{\mathcal{N}_i}(t))$ is used to regulate deviations from the nominally planned trajectory.
A simple linear feedback controller stabilizing the error dynamics \eqref{eq:error_dynamics} can be obtained in a distributed manner using, e.g., the methods introduced in \cite{Conte2016}.
\begin{remark}
	It is possible to design a nonlinear tube controller $\pi_{i}$, such as, e.g., a linear feedback controller with saturation, allowing for the treatment of hard input constraints (e.g., due to physical actuator limits, see also \cite{Hewing2019}).
\end{remark}

In Section \ref{sec:dsmpc}, we introduce a recursively feasible DSMPC scheme based on an indirect feedback formulation~\cite{Hewing2018}.
In Section~\ref{sec:tightening}, we then detail the design of tightened constraints on the nominal system states and inputs, which is performed in an optimization-free and distributed manner, and ensures closed-loop chance constraint satisfaction.

\section{DISTRIBUTED STOCHASTIC MODEL PREDICTIVE CONTROL} \label{sec:dsmpc}

We aim to solve the stochastic control task over the task horizon $\bar{N}$ by employing a receding horizon control formulation over a shortened horizon $N$, i.e.,\ the problem is repeatedly solved at each time step based on the currently measured state at time step $t$, i.e., $x_i(t)$.
Most commonly, robust and stochastic MPC schemes that are based on the separation into a nominal and an error system as in~\eqref{eq:split_dynamics} initialize the nominal dynamics with the currently measured state $x_i(t)$.
In the stochastic setting, this can lead to feasibility issues, in particular due to the potentially unbounded nature of the stochastic disturbance~\cite{Hewing2018e}.
Here, we rely on an indirect feedback stochastic MPC formulation~\cite{Hewing2018}, resulting in the following DSMPC problem:
\newpage

\begin{subequations}
	\begin{align}
	\min_{\textbf{v}} & \sum_{i=1}^{M}\mathbb{E}_{W_i(t)}\left(l_f(x_i(N|t)) + \sum_{k=0}^{N-1}l_{t+k}(x_i(k|t),u_i(k|t))\right) \label{eq:dsmpc_objective} \\
	\text{s.t. } & \forall i \in \mathcal{M}: \nonumber \\
	& x_i(0|t) = x_i(t), \, z_i(0|t) = z_i(1|t\!-\!1), \,  e_i(0|t) = e_i(t) \label{eq:dsmpc_initialization}\\
	& z_i(N|t) = 0 \\
	& W_i(t) = \left[w_i(0|t)^{\top}, \ldots, w_i(N|t)^{\top}\right]\sim\mathcal Q_i(t) \label{eq:dsmpc_disturbance}\\
	& \forall k \in \{0, \ldots, N - 1\}: \nonumber \\
	& \quad z_i(k+1|t) = A_{\mathcal{N}_i}z_{\mathcal{N}_i}(k|t) + B_i v_i(k|t) \\
	& \quad x_i(k+1|t) = z_i(k+1|t) + e_i(k+1|t) \label{eq:dsmpc_predicted_state} \\
	& \quad e_i(k+1|t) = A_{\mathcal{N}_i}e_{\mathcal{N}_i}(k|t) \nonumber \\
	& \qquad \qquad + B_i\pi_{i}(e_{\mathcal{N}_i}(k|t)) + G_iw_i(k|t) \label{eq:dsmpc_error_dynamics} \\
	& \quad u_i(k|t) = v_i(k|t) + \pi_{i}(e_{\mathcal{N}_i}(k|t)) \label{eq:dsmpc_predicted_input} \\
	& \quad h_{i,j}^{x\top}z_i(k|t) \le 1 - c_{i,j,t+k}^x \forall j \in  \{1, \ldots, n_i^x\} \label{eq:dsmpc_tightened_state} \\
	& \quad h_{i,j}^{u\top}v_i(k|t) \le 1 - c_{i,j,t+k}^u \forall j \in  \{1, \ldots, n_i^u\} \label{eq:dsmpc_tightened_input}
	\end{align} \label{eq:dsmpc_problem}%
\end{subequations}
where $\textbf{v} = \mathrm{col}_{k \in \{0,\ldots,N-1\}}(v(k|t))$ with $v(k|t)=\mathrm{col}_{i \in \mathcal{M}}(v_i(k|t))$.
For $k \in \{0, \ldots , N \}$ and every subsystem $i$, the vector $x_i(k|t) \in \mathbb{R}^{n_i}$ denotes the $k$-steps ahead predicted state computed at time step $t$, and $z_i(k|t)$, $v_i(k|t)$, $e_i(k|t)$ and $u_i(k|t)$ the predicted nominal state, nominal input, error and input, respectively.
The input~\eqref{eq:split_dynamics_controller} applied to system~\eqref{eq:local_system_dynamics} is then defined by the solution $\textbf{v}^*$ as 
\begin{align} v_i(t) = v_i^*(0|t). \label{eq:nominal_input}
\end{align}

In problem~\eqref{eq:dsmpc_problem}, the \emph{nominal} state $z_i(0|t)$ is initialized at each time step $t$ with the first predicted nominal state $z_i(1|t\!-\!1)$ obtained at time step $t\!-\!1$, while the state measurement $x_i(t)$ initializes $x_i(0|t)$ (see \eqref{eq:dsmpc_initialization}).
Note that via the optimization of the objective \eqref{eq:dsmpc_objective} with respect to $x_i(k|t)$, feedback is also introduced on the nominal state evolution $z_i(t)$, hence it is referred to as indirect feedback.

As a result of this initialization, the nominal dynamics in \eqref{eq:nominal_dynamics} are valid in closed-loop operation. Note that this is not so, if $z_i(0|t)$ is optimized, as is often the case in robust tube MPC formulations~\cite{Mayne2005}, or if it is set equal to the measured state $x_i(t)$.
From this nominal state evolution, it follows that the closed-loop error evolves independently of the MPC optimization according to (\ref{eq:error_dynamics}) and can therefore be simulated forward by only having access to samples of the disturbances $w_i(t)$.
This allows us to precompute the error prediction \eqref{eq:error_dynamics} prior to solving the optimization problem \eqref{eq:dsmpc_problem}.

Problem \eqref{eq:dsmpc_problem} makes use of tightened constraints on the nominal state and input of each subsystem in \eqref{eq:dsmpc_tightened_state} and \eqref{eq:dsmpc_tightened_input} to realize the chance constraints in \eqref{eq:constraints}.
While the local error feedback $\pi_i$ aims at reducing deviations from the nominally planned trajectory $z_i(t)$, the unknown disturbances $w_i(t)$ cause a non-vanishing error $e_i(t)$ for all $t\geq 0$, which can cause closed-loop constraint violations, even if $h_{i,j}^{x\top}z_{i}(t) \le 1$ and $h_{i,j}^{u\top}v_i(t) \le 1$ holds.
Similar to ideas from robust MPC, we therefore introduce tightened half-space constraints using suitable tightening values $c_{i,j,t}^x$ and $c_{i,j,t}^u$.
In Section~\ref{sec:tightening}, we introduce a data-driven and distributed method to compute $c_{i,j,t}^x$ and $c_{i,j,t}^u$ for the entire task horizon $\bar{N}$ depending on the distribution of the trajectories of the error system~\eqref{eq:error_dynamics}, such that the chance constraints (\ref{eq:constraints}) are fulfilled non-conservatively with the desired probability level.
By non-conservativeness, we refer to the fact that if a given initial condition and any disturbance realization lead to a nominal constraint being active, the corresponding true chance constraint \eqref{eq:constraints} is violated exactly with the specified probability level.

The expectation in the objective \eqref{eq:dsmpc_objective} is taken with respect to a disturbance sequence $W_i(t)$ over the prediction horizon distributed according to $\mathcal Q_i(t)$ (\ref{eq:dsmpc_disturbance}).
For disturbances correlated in time, this addresses the fact that past disturbances provide information which can be utilized in the optimization of the cost.
This information can be used, e.g.,\ by considering the marginal disturbances for each agent conditioned on past disturbance realizations
$$p(W_i(t)) = p\left(W_i(t)|\left[w_{\mathcal{N}_i}(0)^{\top}, \ldots, w_{\mathcal{N}_i}(t - 1)^{\top}\right]^{\top}\right).$$%
The expectation in \eqref{eq:dsmpc_objective} can be evaluated for the special case of i.i.d. disturbances and quadratic costs by considering only the mean of the predicted state and input, see \cite{Hewing2018}.
For a general cost, it is not possible to analytically evaluate the expectation in (\ref{eq:dsmpc_objective}), but it can be approximated based on $N_{s,i}^{MPC}$ samples of the disturbance sequence $W_i(t)$ over the prediction horizon N for each subsystem $i$.
The number of samples trade off prediction accuracy against online computational complexity.

The optimization problem in (\ref{eq:dsmpc_problem}) can be solved in a distributed manner using distributed optimization techniques, see e.g., \cite{Bertsekas1989,Rostami2017}, since the objective and constraints are only coupled between neighboring subsystems. This results in a fully distributed offline and online procedure.

\begin{remark}
	For simplicity, we use a terminal equality constraint in~\eqref{eq:dsmpc_problem}.
	A less restrictive terminal constraint as similarly proposed in \cite{Hewing2019} could be integrated by using a distributed robust positive invariant terminal set, e.g., based on the results introduced in \cite{Conte2013}.
\end{remark}

Recursive feasibility of the distributed MPC scheme (\ref{eq:dsmpc_problem}) can be directly established using results from standard nominal MPC, because the stochastic variables only affect the objective of problem (\ref{eq:dsmpc_problem}) and the constraints are on the nominal states and inputs.
\begin{theorem}\label{thm:recursive_feasibility}
	If the optimization problem (\ref{eq:dsmpc_problem}) is feasible for $x_i(0) = z_i(0)$, then applying the distributed control input \eqref{eq:split_dynamics_controller} with \eqref{eq:nominal_input} to the dynamic system (\ref{eq:local_system_dynamics}), results in problem~\eqref{eq:dsmpc_problem} being feasible for all time steps	$0 \le t \le \bar{N}-N$.
\end{theorem}
\begin{proof}
	The local constraints in (\ref{eq:dsmpc_problem}) can be combined to constraints on the global system state and the proof follows the standard argument in MPC using the shifted sequence from the previous time step, as similarly shown in \cite{Hewing2019} for the centralized case.
\end{proof}

\begin{remark}
	For the special case of a quadratic stage cost, Gaussian disturbances and a terminal weight satisfying the Lyapunov equation, the asymptotic convergence property shown in \cite{Hewing2018} can be extended to the distributed case.
\end{remark}

\section{Distributed Data-driven constraint tightening}\label{sec:tightening}
\begin{algorithm}[t]
	\caption{Computation of tightening values for all
	subsystems $i$, time steps $t$, and half spaces $j$.}
	\label{alg:constraint_tightening}
	\begin{algorithmic}[1]
		\REQUIRE Chance constraints~\eqref{eq:constraints}, confidence level $\beta$,
		$l=1,2,..,N_s$ samples $W^{(l)}$.
		\ENSURE Tightening values $c_{i,j,t}^x$ and $c_{i,j,t}^u$.
		\FOR{every sample $l=1,2,..,N_s$}
			\STATE $(\{e_i^{(l)}(t)\},\{\pi_i(e_{\mathcal N_i}^{(l)}(t))\})\leftarrow$ distributed simulation of error
				system~\eqref{eq:error_dynamics} and corresponding feedback using disturbances $W^{(l)}$ and initial condition $e_i^{(l)}(0)=0$.
		\ENDFOR
		\FOR{every agent $i=1,..,M$, time step $t=0,..,\bar N$}
			\FOR{every half-space $j=1,..,n_i^x$}
				\STATE Compute $c_{i,j,t}^x$ via Alg. 2 and $h_{i,j}^x,\{e_i^{(l)}(t)\},p_{i,j}^x, \beta$
			\ENDFOR
			\FOR{every half-space $j=1,..,n_i^u$}
				\STATE Compute $c_{i,j,t}^u$ via Alg. 2 and $h_{i,j}^u$,$\{\pi_i(e_{\mathcal N_i}^{(l)}(t))\}$,$p_{i,j}^u$, $\beta$
			\ENDFOR
		\ENDFOR
	\end{algorithmic}
\end{algorithm}
\begin{algorithm}[t]
	\caption{Single half-space tightening computation.}
	\label{alg:half_space_prs}
	\begin{algorithmic}[1]
		\REQUIRE Half-space direction $h\in\mathbb R^q$, samples $\xi^{(l)}\in\mathbb R^q$ with $l=1,..,N_s$,
			probability level $p$, and confidence $1-\beta$
		\ENSURE Tightening value $c$
		\STATE $N_d\leftarrow (1-p)N_s-\sqrt{2(1-p)N_s\ln\left(\frac{1}{\beta}\right)}$ 
		\WHILE{number of $\{\xi^{(l)}\} > N_s - N_d$}
			\STATE discard $\xi^{(l^*)}(t)$ with $l^* \leftarrow \argmax_{l}h^\top\xi^{(l)}$
		\ENDWHILE
		\STATE $c \leftarrow \max_{l}h^\top \xi^{(l)}$
	\end{algorithmic}
\end{algorithm}

In the following, we derive a distributed and data-driven algorithm by extending the centralized version in \cite{Hewing2019}, to obtain the tightening of nominal state and input constraints in (\ref{eq:dsmpc_tightened_state}) and (\ref{eq:dsmpc_tightened_input}), based on scenario rather than distributed numerical optimization.
Scenario optimization (see e.g., \cite{Calafiore2005}, \cite{Campi2011}) allows us to perform the tightening based on $N_s$ samples of the disturbance sequence $W^{(l)}$ with $l \in \{1, \cdots , N_s \}$.
With probability $1- \beta$, the resulting tightening ensures satisfaction of the chance constraints (\ref{eq:constraints}) in a non-conservative manner, meaning that if a constraint on the nominal state $z_i(t)$ or input $v_i(t)$ is always active, the probability of the real state $x_i(t)$ or applied input $u_i(t)$, respectively, violating the constraints is exactly $1\!-\!p_{i,j}^x$ and $1\!-\!p_{i,j}^u$ as specified in \eqref{eq:constraints}.
Thereby, the probability $1-\beta$ is related to the number of considered samples $N_s$ of the disturbance sequence.
Compared to related robust approaches, such as~\cite{Conte2013}, the proposed design procedure avoids the solution of a distributed optimization problem involving bilinear matrix inequalities to determine the constraint tightening, by instead making use of samples of the closed-loop error according to the dynamics~\eqref{eq:error_dynamics}.

Specifically, we compute tightening values $c_{i,j,t}^x$ and $c_{i,j,t}^u$, which ensure that the real local state $x_i(t) = z_i(t) + e_i(t)$ and input $u_i(t) = v_i(t) + \pi_{i}(e_{\mathcal{N}_i}(t))$ satisfy the half-space chance constraints (\ref{eq:constraints}) at the desired probability level if the tightened nominal constraints (\ref{eq:dsmpc_tightened_state}) and (\ref{eq:dsmpc_tightened_input}) are always active.
Therefore, we choose the minimal values $c_{i,j,t}^x$ and $c_{i,j,t}^u$ such that for each constraint $j$, time step $t$ and subsystem $i$
\begin{align*}
&\mathrm{Pr}( h^{x\top}_{i,j} e_i(t) \leq c_{i,j,t}^x) \ge p^x_{i,j},  && j \in  \{1, \ldots, n_i^x\}, \\
&\mathrm{Pr}(h^{u\top}_{i,j}\pi_i(e_{\mathcal N_i}(t))\leq c_{i,j,t}^u)\ge p^u_{i,j}, && j \in  \{1, \ldots, n_i^u\},
\end{align*}
holds, bounding the distribution of the local error dynamics~\eqref{eq:error_dynamics} and error feedback in the local half-space directions~$h_{i,j}^x$ and $h_{i,j}^u$.
The tightening values $c_{i,j,t}^x$ and $c_{i,j,t}^u$ can be obtained by solving the stochastic optimization problems%
\begin{subequations}\label{eq:stochastic_prs_problem}
	\begin{align}
		c_{i,j,t}^x&=\min c_x\text{ s.t. } \Pr(h_{i,j}^{x\top} e_i(t)\leq c_x)\geq p^x_{i,j},\\
		c_{i,j,t}^u&=\min c_u\text{ s.t. } \Pr(h_{i,j}^{u\top} \pi_i(e_{\mathcal N_i}(t))\leq c_u)\geq p^u_{i,j}.
	\end{align}
\end{subequations}
Arguments from scenario optimization allow us to approximate these stochastic optimization problems by sampled versions, where $h_{i,j}^{x\top} e_i^{(l)}(t)\leq c_x$ and $h_{i,j}^{u\top} \pi_i(e_{\mathcal N_i}^{(l)}(t))\leq c_u$ are enforced as deterministic constraints for sampled error trajectories $e_i^{(l)}(t)$ based on disturbance samples $W^{(l)}$ as detailed in Algorithm~\ref{alg:constraint_tightening}.
In fact, scenario-based optimization arguments~\cite{Campi2011} provide a confidence level $1-\beta$ at which the sample-based solution fulfills the probabilistic constraints in~\eqref{eq:stochastic_prs_problem} and even allow us to discard a certain fraction of the most restrictive samples.
The procedure is outlined in Algorithms~\ref{alg:constraint_tightening} and~\ref{alg:half_space_prs}.
In order to increase the confidence level $1-\beta$, a larger number of samples could be considered.

Algorithm~\ref{alg:constraint_tightening} takes the chance-constraints~\eqref{eq:constraints} as inputs, as well as disturbance samples $W^{(l)}$, and the confidence level parameter $\beta$, where $1-\beta$ corresponds to the confidence level of the scenario optimization, i.e., the confidence at which the computed constraint tightening results in closed-loop chance constraint satisfaction.
In a first step, we generate the relevant error scenarios by simulating the error system for each disturbance sample, see Algorithm~\ref{alg:constraint_tightening}, lines 1-3.
Note that the disturbance samples $W^{(l)}$ can be stored distributedly and that the simulation is a distributed operation requiring only neighbor-to-neighbor communication and therefore scales to arbitrarily large networks.
After generating the error scenarios, every agent can approximately solve~\eqref{eq:stochastic_prs_problem} for each state and input half-space separately in lines 5-10 using the subroutine in Algorithm~\ref{alg:half_space_prs}.
In Algorithm~\ref{alg:half_space_prs}, line 1, we first compute the number of scenarios $N_d$ that can be discarded based on the desired probability level $p$ and confidence level $1-\beta$, see \cite{Hewing2019,Campi2011} for details.
To determine the required half-space level $c$, we iterate over the disturbance samples	and discard the $N_d$ most restricting samples in Algorithm~\ref{alg:half_space_prs}, lines 2-4.
The most restrictive remaining disturbance sample is then used to obtain the required tightening value in line 5.
Note that increasing the number of samples $N_s$ either allows us to achieve a higher probability level $p$ for the chance constraints, or a higher confidence level $1- \beta$ of the scenario optimization problem.
Since the required number of samples scales logarithmically with $\beta$, the confidence level can typically be chosen to be very high~\cite{Campi2011}.
Note that the number of samples $N_{s,i}^{MPC}$ chosen to approximate the MPC cost is not related to the number of samples $N_s$ to perform the constraint tightening and does not affect constraint satisfaction guarantees.
In fact, one would typically have $N_s \gg N_{s,i}^{MPC}$ since the number of samples for constraint tightening does not affect the online computation, and the required offline computations are reasonably cheap.

Recursive feasibility of problem (\ref{eq:dsmpc_problem}) as shown in Theorem~\ref{thm:recursive_feasibility} and the tightened constraints on the nominal states (\ref{eq:dsmpc_tightened_state}) and inputs  (\ref{eq:dsmpc_tightened_input}) with constants $c_{i,j,t}^x$ and $c_{i,j,t}^u$ obtained using Algorithm~\ref{alg:constraint_tightening} allow us to establish a guarantee for the satisfaction of the chance constraints (\ref{eq:constraints}) on states $x_i(t)$ and inputs $u_i(t)$ of each subsystem in closed-loop.
\begin{theorem}\label{thm:constraint_satisfaction}
	Let $c_{i,j,t}^x$ and $c_{i,j,t}^u$ be obtained using Algorithm~\ref{alg:constraint_tightening} and the control law \eqref{eq:split_dynamics_controller} with \eqref{eq:nominal_input} be applied to the distributed system (\ref{eq:local_system_dynamics}).
	With probability $1-\beta$, the resulting local states $x_i(t)$ and inputs $u_i(t)$ satisfy the chance constraints in (\ref{eq:constraints}).
\end{theorem}
\begin{proof}
	The proof follows the proof of Theorem 3 in \cite{Hewing2019}, which is summarized here for completeness.
	Algorithm~\ref{alg:constraint_tightening} greedily discards $N_d$ of the initial $N_s$ samples $e_i^{(l)}$ and sets the tightening value $c_{i,j,t}^x$ as the maximum over the remaining samples of $h_{i,j}^{x\top}e_i^{(l)}(t)$ via Algorithm~\ref{alg:half_space_prs} line 5.
	Therefore, for all remaining samples it holds that $h_{i,j}^{x\top}e_i^{(l)}(t) \le c_{i,j,t}^x$.
	From scenario optimization, we then have with probability $1 - \beta$, that $\mathrm{Pr}( h^{x\top}_{i,j} e_i(t) \leq c_{i,j,t}^x) \ge p^x_{i,j}$. Therefore, constraining the local nominal state $z_i(t)$ to the tightened constraints (\ref{eq:dsmpc_tightened_state}) results in the real state of the system $x_i(t) = z_i(t) + e_i(t)$ fulfilling the chance constraints in (\ref{eq:constraints}).
	The same arguments hold for the input constraints by using Algorithm~\ref{alg:constraint_tightening} to obtain the tightening values $c_{i,j,t}^u$.
\end{proof}
\begin{remark}
	The chance constraint satisfaction property in Theorem~\ref{thm:constraint_satisfaction} renders the proposed DSMPC framework suitable for safety certification of distributed learning-based controllers in the line of \cite{Muntwiler2019}, i.e., using a distributed MPC to verify and modify a proposed learning input if necessary.
	While satisfaction of constraints can only be ensured in probability, the computational complexity and conservatism can be dramatically reduced compared with other distributed safety certification schemes, see, e.g., \cite{Muntwiler2019}, \cite{Larsen2017}. 
\end{remark}

\begin{remark}
	For disturbances with zero mean and known variance, e.g., $W \sim \mathcal{N}(0,\Sigma_W)$ and a distributed linear tube control law $\pi_i(e_{\mathcal{N}_i}(t)) = K_ie_{\mathcal{N}_i}(t)$ with $K_i \in \mathbb{R}^{n_i \times n_{\mathcal{N}_i}}$, one can analytically compute the mean and variance of the error sequence. Instead of a data-based tightening, an analytic tightening is then possible using the marginal local and neighborhood variances $\Sigma^e_i(t)$ and $\Sigma^e_{\mathcal{N}_i}(t)$, e.g.,\ as
	\begin{align*}
		c_{i,j,t}^x &= \phi^{-1}\!(p^x_{i,j}) \sqrt{h^\top \Sigma^e_i(t) h} ,&& j \in  \{1, \ldots, n_i^x\}, \\
		c_{i,j,t}^u &=\phi^{-1}\!(p^u_{i,j}) \sqrt{h^\top K_i\Sigma^e_{\mathcal{N}_i}(t)K_i^\top h}, && j \in  \{1, \ldots, n_i^u\},
	\end{align*}
	where $\phi^{-1}$ is the quantile function of the standard normal distribution and all computations can be easily carried out in a distributed manner.
	A related approach computing the full variance matrix was presented in \cite{Mark2019}, where, using a possibly conservative additional step, guarantees are given for all subsystems simultaneously.
\end{remark}

\section{SIMULATION EXAMPLE} \label{sec:simulation_example}

To highlight the effectiveness of the proposed DSMPC scheme we consider the example of a distributed cooling system as used in \cite{Recht2019}.
The task of the cooling system is to control the temperature of a server farm, which can similarly be interpreted, e.g., as the temperature of production machines in a big manufacturing plant.
Each local subsystem thereby has a heat source (e.g., heat production due to the computational load) and a cooling component (e.g., a fan or water cooling system).
The temperature of each subsystem affects the temperature of neighboring systems.
Cooling of the system is important in order to prevent defects due to overly high temperatures or safety shutdowns.
At the same time, excessive cooling should be prevented.

We consider a server farm with $M=100$ servers arranged on a regular $10$ by $10$ grid with equal spacing $r$.
Each server is thermally coupled with its direct neighbors in the grid.
The servers heat up due to their computational load, and their temperature influences that of neighboring servers.
Disturbances acting on each local server mimic the temperature increase (or decrease) due to high (or low) computations compared to the average computational load acting on the servers.
The computational load is assumed to have a known time-varying mean over the course of the day.

The local system dynamics are defined as 
\begin{equation}\label{eq:example_system_dynamics}
	x_i(t\!+\!1) = 1.01 x_i(t) + \sum_{j \in \mathcal{N}_i \backslash i}\frac{0.01}{1 + r}x_j(t) + u_i(t) + w_i(t),
\end{equation}
where $x_i(t)$ denotes the deviation from a desired temperature of operation $\bar{T}_i = 25^{\circ}C$, with the actual temperature $T_i(t) = \bar{T}_i + x_i(t)$, $u_i(t)$ denotes a local cooling input and $w_i(t)$ the disturbance acting on each server.
The local disturbances are modeled as a non-i.i.d. multivariate Gaussian distribution with a sinusoidal mean vector $\mu \in \mathbb{R}^{\bar{N}}$ and correlation in time over the task horizon due to the convariance matrix $\Sigma \in \mathbb{R}^{\bar{N} \times \bar{N}}$.
We introduce state and input constraints as 
\begin{subequations}
	\begin{align}
		-5 &\le x_i \le 5, \quad \forall i \in \mathcal{M}, \\
		-1 &\le u_i \le 1, \quad \forall i \in \mathcal{M},
	\end{align}
\end{subequations}
with desired probability level of $0.9$.
We use a tube controller $\pi_{i}(e_i(t)) = -0.5 e_i(t)$ for every subsystem resulting in the closed-loop nominal error system
\begin{equation}
	e_i(t\!+\!1) = 0.51 e_i(t) + \sum_{j \in \mathcal{N}_i \backslash i}\frac{0.01}{1 + r_{ij}}e_j(t),
\end{equation}
which is stable according to the Gershgorin Circle Theorem~\cite{Gerschgorin1931} if for all subsystems $i$
\begin{equation}
	0.51 + \sum_{j \in \mathcal{N}_i \backslash i}\frac{0.01}{1 + r_{ij}} < 1.
\end{equation}

We simulate the behavior of the DSMPC scheme for the system \eqref{eq:example_system_dynamics} with sampling time of $0.5h$ and prediction horizon of $N = 12h$ over an effective task horizon $\bar{N}-N = 2d$.
The constraints are tightened using $N_s = 100$ disturbance samples for each subsystem. The samples for subsystem~9 are shown in the third subplot of Fig.~\ref{fig:evolution_subsystem_9}.
The local stage costs are assumed to have the form
\begin{equation}
	l_t(x_i,u_i) = x_i^{\top}x_i + 1000 u_i^{\top}u_i,
\end{equation}
which represents high cooling costs. The MPC cost is approximated using $N_{s,i}^{MPC} = 10$ samples for each subsystem.

Fig. \ref{fig:evolution_all_systems} shows the temperature evolution and corresponding inputs over the course of two days.
While the chance constraints on the states are violated for only four subsystems, the input constraints always hold.
The upper two subplots of Fig. \ref{fig:evolution_subsystem_9} show temperature and inputs of subsystem $9$ including the nominal states and inputs and the respective time-varying nominal constraints.
The constraints on the nominal state and input are active at several instances in time.
\begin{figure}
	\centering
	\includegraphics[width=\columnwidth]{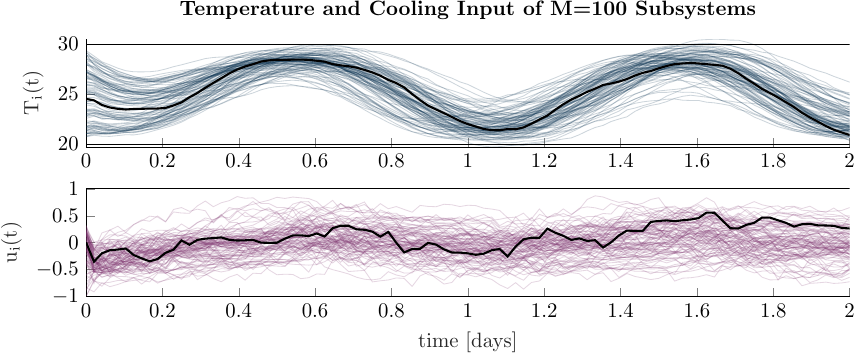}
	\caption{Evolution of the temperature and control input for all subsystems over time with the one of subsystem $9$ indicated in black.}
	\label{fig:evolution_all_systems}
\end{figure}
\begin{figure}
	\centering
	\includegraphics[width=\columnwidth]{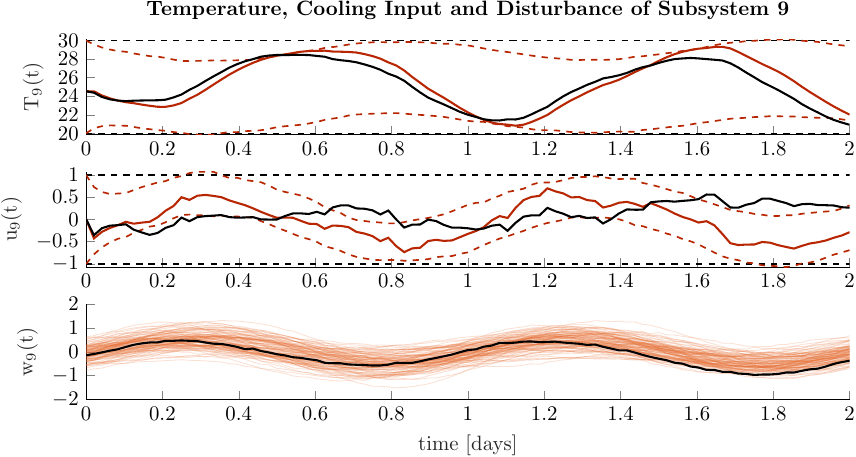}
	\caption{The upper two subplots show the evolution of the real temperature and input (in black), nominal temparature and input (in red) and real and tightened constraints (dashed lines) for subsystem 9.
	The third subplot shows the disturbance samples used for tightening the constraints (in orange) and the actual disturbance acting on the system (in black).}
	\label{fig:evolution_subsystem_9}
\end{figure}

\section{CONCLUSIONS} \label{sec:conclusion}

In this paper, we introduced a distributed stochastic MPC framework that ensures recursive feasibility, based on an indirect feedback formulation, and satisfaction of chance constraints in closed-loop in a non-conservative manner due to a data-driven and optimization-free constraint tightening approach.
Both, the offline controller synthesis as well as the online operation can be performed in a completely distributed manner, offering a scalable and high performance scheme with reduced conservatism compared with the literature.


\bibliography{bibliography}

\end{document}